\documentclass{ieeeconf}
\IEEEoverridecommandlockouts
\usepackage{verbatim}
\usepackage{epsfig}
\usepackage{amsmath}
\usepackage{amsthm}
\usepackage{amssymb}
\usepackage{psfrag}
\usepackage[T1]{fontenc}
\usepackage[ruled]{algorithm2e}

\usepackage{url}
\usepackage{dsfont}
\usepackage{mathtools}

\usepackage[usenames,dvipsnames]{pstricks}
\usepackage{pst-grad} % For gradients
\usepackage{pst-plot} % For axes
\usepackage[ruled]{algorithm2e}
\usepackage{etoolbox}

\usepackage{graphicx}
\usepackage[font=footnotesize,skip=0pt]{caption}

\makeatletter
\patchcmd{\@makecaption}
  {\scshape}
  {}
  {}
  {}
\makeatother

\usepackage[noadjust]{cite}

\newtheorem{theorem}{Theorem}[]
\newtheorem{lemma}{Lemma}[]

\newtheorem{assmp}{Assumption}

\newcommand{\G}{\mathcal{G}}
\newcommand{\V}{\mathcal{V}}
\newcommand{\E}{\mathcal{E}}

\newcommand{\A}{\mathbf{A}}

\newcommand{\I}{\mathbf{I}}

\newcommand{\uu}{\text{u}}
\newcommand{\yy}{\text{y}}

\newcommand{\R}{\mathbb{R}}
\newcommand{\rr}{\mathbf{r}}

\newcommand{\du}{\mathbf{d_u}}
\newcommand{\dy}{\mathbf{d_y}}
\newcommand{\dtheta}{\mathbf{d}_{\theta}}

\newcommand{\ones}{1}

\newcommand{\p}{\mathbf{P}}
\newcommand{\yb}{\mathbf{y}}
\newcommand{\x}{\mathbf{x}}
\newcommand{\ka}{\mathcal{M}}
\newcommand{\rc}{\mathcal{R}}

\DeclareMathOperator*{\Proj}{Proj}
\DeclareMathOperator*{\vect}{vec}

\newcommand{\nul}{\operatorname{null}}
\newcommand{\sn}{\operatorname{span}}

\DeclareMathOperator{\zb}{\mathbf{z}}

\DeclareMathOperator{\N}{\mathcal{N}}

\DeclareMathOperator*{\argmin}{argmin}
\DeclareMathOperator*{\minimize}{minimize}

\newcommand\norm[1]{\left\lVert#1\right\rVert}

\newcommand{\longthmtitle}[1]{\mbox{}{\bf \textit{(#1).}}}

\IEEEpeerreviewmaketitle

\begin{document}

\title{\textbf{Adaptive Online Model Update Algorithm for Predictive Control in Networked Systems}
% A Distributed Model Identification Algorithm for Multi-Agent Systems
%Distributed Model Identification Using Input-Output Data
}
\author{Vivek Khatana, Chin-Yao Chang, Wenbo Wang
% \thanks{This work is supported by the Advanced Research Projects Agency-Energy OPEN through the project titled "Rapidly Viable Sustained Grid" via grant no. DE-AR0001016.}
\thanks{ Vivek Khatana is with the Department of Electrical and Computer Engineering, University of Minnesota, Minneapolis, USA (Email: \{khata010\}@umn.edu).}%
\thanks{ Chin-Yao Chang and Wenbo Wang are with the National Renewable Energy Laboratory, Golden, CO 80401, USA (Email: \{chinyao.chang, wenbo.wang\}@nrel.gov).}
% \thanks{$^{1}$Muti V. Salapaka are with Department of Electrical and Computer Engineering, University of Minnesota, Minneapolis, USA,
% {\tt\small murtis@umn.edu}}
\thanks{This work was authored in part by NREL, operated by Alliance for Sustainable Energy, LLC, for the U.S. Department of Energy (DOE) under Contract No. DE-AC36-08GO28308. Funding provided by DOE Office of Electricity, Advanced Grid Modeling Program, through agreement NO. 33652. The views expressed in the article do not necessarily represent the views of the DOE or the U.S. Government. The U.S. Government retains and the publisher, by accepting the article for publication, acknowledges that the U.S. Government retains a nonexclusive, paid-up, irrevocable, worldwide license to publish or reproduce the published form of this work or allow others to do so, for the U.S. Government purposes.}
}

% make the title area
\maketitle

\begin{abstract}
In this article, we introduce an adaptive online model update algorithm designed for predictive control applications in networked systems, particularly focusing on power distribution systems. Unlike traditional methods that depend on historical data for offline model identification, our approach utilizes real-time data for continuous model updates. This method integrates seamlessly with existing online control and optimization algorithms and provides timely updates in response to real-time changes. This methodology offers significant advantages, including a reduction in the communication network bandwidth requirements by minimizing the data exchanged at each iteration and enabling the model to adapt after disturbances. Furthermore, our algorithm is tailored for non-linear convex models, enhancing its applicability to practical scenarios. The efficacy of the proposed method is validated through a numerical study, demonstrating improved control performance using a synthetic IEEE test case.
\\[1ex]
\textit{keywords}: Model-identification, data-driven model predictive control, distributed optimization, online optimization, power grid, networked systems. 
\end{abstract}

\section{Introduction}\label{sec:introduction}
System models are essential for understanding and controlling complex systems, as they enable accurate predictions, analysis, and optimization of resources. The mathematical models abstract complex, real-world phenomena into manageable representations. With the increase in data availability, scientists and practitioners favor adaptive reconfiguration of the model representations to conform to the latest data measurements. In this article, we focus on identifying the sub-system models capturing the global behavior of a networked system to solve the following predictive control problem
\begin{align}\label{eq:control_prob}
\minimize_{\uu(1), \uu(2), \dots, \uu(T) \in \mathbb{R}^\du} & \sum_{k=1}^T \sum_{i=1}^N \ell_{i}(y(k),k) + h_{i}(\uu_i(k), k) \\
\mbox{subject to} \ \uu_i(k) & \in \mathcal{U}_i, \ \mbox{for all} \ i = 1,2,\dots, N, \ \mbox{for all} \ k  \nonumber  
\end{align}
where, $\uu(k) = [\uu_1(k), \uu_2(k), \dots, \uu_N(k)] \in \mathbb{R}^\du$ is the control decision at time $k$ with constraint sets $ \mathcal{U}_i \subseteq \mathbb{R}^{\du_i}$, $\sum_{i=1}^N \du_i = \du$,
% $\mathcal{U}:= \mathcal{U}_1 \times \cdots \times \mathcal{U}_N$,
and $y(k) \in \mathbb{R}^\dy$ is an observable that involves the physical or behavioral inter-dependencies among the sub-systems. Functions $\ell_{i}$ and $h_{i}$ in~\eqref{eq:control_prob} capture the costs due to the output $y(t)$ and the control inputs $\uu_i$ respectively at the sub-system $i$. The predictive control problems~\eqref{eq:control_prob} appear in the context of online optimal control, communication systems, and robotic networks \cite{roque2022coordination,spudic2015cooperative} to mention a few. More recently, the problem has also been of interest in the control and operation of power systems \cite{andrey_2}.

Suppose, the sub-system inter-dependencies are modeled via the parametric description between the observable $y$ and the inputs $\uu$,  
\begin{align}
   \ka(\theta) :  y_\theta(t) = \phi(\uu(t), \theta), \label{eq:data_model}
\end{align}
where $\theta \in \mathbb{R}^{\mathbf{d}_\theta}$ is the parameter of the map $\phi: \mathbb{R}^{\du} \times \mathbb{R}^{\mathbf{d}_\theta} \to \mathbb{R}^{\dy}$ that captures the relation between the control inputs, $\uu(t)$, and the (parameterized) observable or output of
the networked system, $y_\theta(t)$, at time $t$. It is assumed that there exists a vector $\theta^\star$ such that the true system output $y(t)$ is given by $y(t) = y_{\theta^\star}(t) = \phi(\uu(t),\theta^\star)$. Note that optimization~\eqref{eq:control_prob} presumes the knowledge of the input-output map~\eqref{eq:data_model}. The decisions $\uu(k)$ are determined based on the postulated output $y_\theta (k)$ via a model of the form~\eqref{eq:data_model} and are sensitive to model mismatches. Under model imperfections, the generated control inputs might drive the network operation to an undesirable state. 

Given~\eqref{eq:control_prob} and~\eqref{eq:data_model}, the problem addressed in the current article pertains to the development and analysis of an algorithm that enables the update of the parametric input-output map in an online manner based on current data measurements to incorporate real-time variations of the controlled system and generate optimal control inputs. The exact description of the class of parametric maps chosen is given in Section~\ref{sec:dis_model_ID}. 

\subsection{Literature Review}
System identification~\cite{ljung1998system} is a broad topic that spans multiple fields. Specialized methodologies, such as learning-based methods~\cite{hastie2009overview} and behavioral system theory for non-parametric models~\cite{markovsky2021behavioral}, can generally be viewed as system identification. As for the applications for power distribution systems, utilities typically maintain feeder models in distribution planning and geographic information system databases~\cite{feeder_model}. However, operational changes to the grid, such as upgrades and reconfigurations~\cite{reconfiguration}, as well as database errors~\cite{data_issue}, necessitate ongoing maintenance of these models and databases~\cite{EPRI_distribution_guidelines}. Voltage control and other operational controls~\cite{bernstein2023time,bernstein2019online}, if based on erroneous or outdated models and data, can adversely affect system stability and reliability. Model identification techniques can address these model consistency issues. Some approaches involving machine learning methods~\cite{chiuso2019system,alimi2020review} require centralized data collection, which raises data privacy concerns and necessitates communication infrastructures. A multi-agent-based distributed approach~\cite{mcarthur2007multi,mahela2020comprehensive} can mitigate these concerns. Building on this foundation, our previous work,~\cite{chang2023privacy}, advanced the state-of-the-art in distributed identification methods that protect local data privacy for \textit{linear} systems, albeit with some communication requirements.

Building on the merits of our prior research's distributed and localized approach, the current article introduces several advancements. The main contributions are as follows:
\begin{enumerate}
\item  We develop a distributed algorithm for online model identification in networked nonlinear systems to enable the observable estimate in problem~\eqref{eq:control_prob} to align with the true output of the system. 
\item We establish that the proposed online algorithm has a sublinear regret in identifying the true convex input-output map of the system.
\item The developed algorithm has several desirable properties:
\begin{itemize}
    \item it preserves the local input data privacy for every sub-system.
    \item it requires only the latest measurements for updates, eliminating the need for storing historical data, and has substantially less communication bandwidth requirement. 
\end{itemize} 
\end{enumerate}
We present a numerical simulation study to demonstrate the performance of the proposed algorithm. The predictive control problem~\eqref{eq:control_prob} is instantiated as a voltage regulation problem in power systems. The numerical results corroborate the efficacy of the proposed algorithm in adaptively updating the input-output map of the test power system utilizing the latest measurements. The results establish that having access to an accurate nonlinear model provided by the proposed algorithm results in superior control performance compared to traditional linear models for power distribution systems, underscoring the practical value of our framework.\\[1ex]
At this point, we emphasize that the current work is not related to the body of research on the reconstruction and identification of unknown topology of an interconnected system using time-series measurement data \cite{mvsalapaka, weerts2018identifiability, mishfad_work}.\\[1ex]
The rest of the paper is organized as follows: Section~\ref{sec:def_notations} introduces some key definitions and notations used throughout the article. Sections~\ref {sec:dis_model_ID} and~\ref{sec:model_update} delve into the distributed system model identification framework and provide the preliminary analysis to aid the development of the model identification algorithm. 
Section~\ref{sec:sysid_alg} presents the proposed  algorithm, its convergence analysis and the distributed implementation details. A predictive control problem with online model updates is presented in Section~\ref{sec:online_controller}. Section~\ref{sec:num_sim} provides the simulation study and demonstrates the numerical results on the performance of the developed algorithm in solving the predictive control problem with online model updates introduced in Section~\ref{sec:online_controller}. The concluding remarks are provided in Section~\ref{sec:conclusion} with some directions for future research.

\subsection{Definition and Notations} \label{sec:def_notations}
In this paper, we denote matrices in boldface. 
For a vector $x\in\R^n$, we denote its $\ell^2$-norm and the norm induced by a matrix $\A\succ 0$ by $\|x\|_2$ and $\|x\|_\A$, respectively. 
%Let $\diag{(x)}$ represent the matrix whose diagonal elements are the elements of the vector $x \in \R^n$. 
%For matrices $\A^1,\A^2,\cdots,\A^N \in \mathbb{R}^{m \times n}$, we denote $\blkdiag{(\{\A^i\}_{i=1}^{N})}$ as the block diagonal matrix of the matrices $\A^i$. 
The vertical and horizontal concatenation of matrices $\A^i$ are denoted as $[\A^1;\A^2;\cdots; \A^n]\in\R^{Nm \times n}$ and $[\A^1,\A^2,\cdots, \A^n]\in\R^{m\times Nn}$. The  Kronecker product of matrices $\A^i$ and $\A^j$ is denoted as $\A^i \otimes \A^j$.
For a matrix $\A\in\R^{m\times n}$, $\vect(\A)\in\R^{mn}$ is a column vector created by concatenating the column vectors of $\A$ from left to right. For a matrix $\A \in \R^{m \times n}$, $\nul \{\A \} := \{x \in \R^n | \A x = 0\}$ denotes the null space of matrix $\A$. Given a set $S$ of vectors, the linear span of $S$ is defined as $\sn \{ S \} := \{\sum_{i=1}^N v_i x_i | v_i \in \R, x_i \in S \}$. 
The scalar element of the $i^{th}$ row and $j^{th}$ column of $\A$ is denoted as $\A_{i}^j$ and the $j^{th}$ row and column of the matrix $\A$ are denoted as $\A_{j,:}$ and $\A_{:,j}$, respectively. The identity matrix and vector with all entries equal to $1$ of dimension $n$ are denoted as $\I_n$ and $\ones_n$, respectively. 

A graph $\mathcal{G}$ is denoted by a pair $(\V,\E)$ where $\V$ is a set of vertices (or nodes) and $\E$ is a set of edges, which are ordered subsets of two distinct elements of $\V$. If an edge from $j \in \V$ to $i \in \V$ exists then it is denoted as $(i,j)\in \E$. The set of neighboring sub-systems of node $i \in \mathcal{V}$ is called the neighborhood of node $i$ and is denoted by $\N_i = \{j \ | \ (i,j)\in \mathcal{E}\}$. In the subsequent, we use the terms agents, nodes, and sub-systems interchangeably. A continuous function $f: \mathbb{R}^p \to \mathbb{R}$ is called Lipschitz continuous with constant $L > 0$ if the following inequality holds:  $ |f(x) - f(y)| \leq L \|x-y\|, \ \forall \ x,y \in \mathbb{R}^p$. 
Given a norm $\|\cdot \|$ and a set $K \subset \mathbb{R}^p$, define the diameter of $K$ with respect to this norm as $Diam_{\|\cdot \|}(K) := \sup_{x, y \in K} \|x-y\|$. In the subsequent text the $O(.)$ and $o(.)$ operations denote the standard \textit{Big-O} and \textit{Little-o} notations respectively \cite{knuth1997art}. 

\section{Agent Based System Framework} \label{sec:dis_model_ID}
We consider a networked system represented by a graph $\G(\V,\E)$ consisting of $|\V|:= N$ nodes (or sub-systems). Each node $i$ has an actuator applying the control decision $\uu_i \in \mathbb{R}^{\du_i}$. Assume that $\dy$ number of sensors are deployed in the network $\G$. The measurements of all these sensors are sent to a fusion center that collects all the sensor outputs to create a measurement $\widehat{y}(t) \in \mathbb{R}^\dy$ of the true global observable $y(t)$. Every sub-system $i$ maintains a local estimate $\widehat{y}^i_{\theta_i}(t)$ of the global observable $y(t)$ via a model of the kind in~\eqref{eq:data_model}. In particular, 
\begin{align}\label{eq:individual_estimate}
    \widehat{y}^i_{\theta_i}(t) :=  \phi_i(\uu_i(t), \theta_i), \ \mbox{for all} \ i = 1,2,\dots, N,
\end{align}
where parameter $\theta_i \in \mathbb{R}^{\mathbf{d}_{\theta_i}}$, with $\sum_{i=1}^N \mathbf{d}_{\theta_i} = \mathbf{d}_{\theta} $. The output estimate of the network is defined as
\begin{align}\label{eq:aggregate_output}
    \widehat{y}_{\theta}(t) = \frac{1}{N} \sum_{i=1}^N \widehat{y}^i_{\theta_i}(t) = \frac{1}{N} \sum_{i=1}^N \phi_i(\uu_i(t), \theta_i).
\end{align}
Here we assume that the parametric model above is accurate
in the sense that there exists a $\theta^\star := [\theta_1^\star; \theta_2^\star, \dots; \theta_N^\star] \in \mathbb{R}^{\mathbf{d}_\theta}$ such that 
\begin{align}\label{eq:avg_model}
    \widehat{y}(t) = \widehat{y}_{\theta^\star}(t) = \frac{1}{N} \sum_{i=1}^N \phi_i(\uu_i(t), \theta^\star_i), \ \mbox{for all} \ t.
\end{align}
With~\eqref{eq:avg_model}, we formulate the following predictive control problem,
\begin{align}\label{eq:control_avg_model}
\minimize_{\uu(1), \uu(2), \dots, \uu(T) \in \mathbb{R}^\du} & \sum_{k=1}^T \sum_{i=1}^N \ell_i(\widehat{y}_\theta(k),k) + h_i(\uu_i(k), k) \\
\mbox{subject to} \ \uu_i(k) & \in \mathcal{U}_i, \mbox{ for all} \ i = 1,2,\dots, N, \ \ \mbox{for all} \ k. \nonumber   
\end{align}
Note that~\eqref{eq:control_avg_model} is equivalent to problem~\eqref{eq:control_prob} if $y(t) = \widehat{y}(t) = \widehat{y}_\theta(t)$ hold for all $t$. This is typically assumed in the state-of-the-art to solve problem~\eqref{eq:control_prob} (see \cite{bernstein2023time}, Assumption~4, \cite{bernstein2019online}, Assumption~5, for example). 
However, when the output estimate $\widehat{y}_\theta$ doesn't match with the true measurements of the system the control performance is affected adversely. In this article, we take the approach of adaptively improving the parametric output estimate of the sub-systems to maintain the validity of~\eqref{eq:avg_model}. Denote
$L(\theta) = |\sum_{k=1}^T \sum_{i=1}^N \ell_i(\widehat{y}_\theta(k),k) - \sum_{k=1}^T \sum_{i=1}^N \ell_i(\widehat{y}(k),k) |$ as the cost of model mismatch with the controller running over a horizon $T$. Our goal is to develop algorithms that minimize the model mismatch quantified by $L(\theta)$ in real-time so that the performance of the closed-loop controllers in solving~\eqref{eq:control_avg_model} is not compromised due to model mismatches.

%We formally state the problem as follows: Consider that the networked system $\mathcal{S}$ is operating in closed-loop with a controller solving problem~\eqref{eq:control_avg_model} with the estimated output generated via some parameters $\theta^0$. Initially, the controller is such that the closed-loop performance is satisfactory, that is given $\varepsilon > 0, H(\theta^0) \leq \varepsilon$. At some point, it is determined that due to a mismatch between the model's estimated output and the true output, the closed-loop performance is no longer satisfactory. Therefore, the input-output map needs to be improved to restore the closed-loop performance.

\section{Model Update Problem and the Distributed Reformulation}\label{sec:model_update}
Given the criticality of~\eqref{eq:avg_model} we aim to reduce the model mismatch by finding the parameter $\theta$ that solves
% \begin{align}
%     \rr(\theta) = \frac{1}{2}\sum_{t = 1}^T \norm{ y(t) - g (u(t), \theta) }^2.
% \label{eq:original_problem}
% \end{align}
\begin{align}
   \hspace{-0.08in} \minimize_{\theta} \ \rr(\theta) := \frac{1}{2}\sum_{t = 1}^T \norm{ \widehat{y}(t) - \frac{1}{N} \sum_{i=1}^N \phi_i(\uu_i(t), \theta_i) }^2.
\label{eq:original_problem}
\end{align}
\noindent For the subsequent development, we make the following assumptions: 

\begin{assmp}\label{assmp:bibo}
    The control decisions generated via problem~\eqref{eq:control_avg_model} ensure the stability of the networked system $\mathcal{S}$. 
\end{assmp}

\begin{assmp}\label{assmp:separable_model}
Functions $\phi_{i}$ in~\eqref{eq:individual_estimate} are proper, convex, and Lipschitz continuous with constant $L_{i}$ for all $i$.
    % Functions $\phi_{i}(., \mu)$ are proper, convex, and Lipschitz continuous in  $\mu$ with constant $L_{i}$ for all $i$.
\end{assmp}

% $:= [a_{ii_1} \phi_i^{i_1}(\uu_i,\theta_i^{i_1}), a_{ii_2} \phi_i^{i_2}(\uu_i,\theta_i^{i_2}), \dots, $ $ a_{ii_J} \phi_i^{i_J}(\uu_i,\theta_i^{i_J})] = [g_{i}^{i_1} (\uu_i,\theta_i^{i_1}),g_{i}^{i_2}(\uu_i,\theta_i^{i_2}),\dots,  g_{i}^{i_1}(\uu_i,\theta_i^{i_J})]$ $ = g_i (\uu_i(t), \theta_i) \in \mathbb{R}^{\dy \times i_J}$. 
% $\theta_i, \phi_i^a$, and $\widehat{y}^i_{\theta_i}(t)$ we have,  $\widehat{y}^i_{\theta_i}(t) = \phi_i^{a}(\uu_i(t), \theta_i) \mathbf{1}_{i_J}$. 

\noindent % Define $\phi^a_i (\uu_i(t), \theta_i) := a_{i}\phi_i(\uu_i (t),\theta_i)$. 
Note that under the network model, each sub-system $i$ creates an estimate $\widehat{y}^i_{\theta_i}(t)$ to determine the effect of its regional control decisions $\uu_i(t)$ on the output (reflected in the measurements). Problem~\eqref{eq:original_problem} can be interpreted as a distributed optimization problem across a network of sub-systems as:
\begin{align}
   & \minimize_{\theta} \ \frac{1}{2N^2}\sum_{t = 1}^T \norm{\sum_{i=1}^N \left( \widehat{y}(t) -  \phi_i(\uu_i(t), \theta_i) \right) }^2.
\label{eq:intermediate_opt_problem}
\end{align}
where we rewrite $\widehat{y}$ by $\frac{1}{N}\sum_{i=1}^N\widehat{y}$ in~\eqref{eq:original_problem} to derive~\eqref{eq:intermediate_opt_problem}.
% \subsection{Distributed reformulation of the system modeling}\label{sec:dis_reform}
% In this section, we go through a series of reformulations of~\eqref{eq:data_model} for convenience of distributed algorithm design. With Assumption~\ref{assmp:separable_model}, we have the following formulation of $\rr(\theta)$:
% \begin{align}\label{eq:first_reformulation}
%    \rr(\theta) = \frac{1}{2}\sum_{t = 1}^T \norm{ y(t) - \sum_{i=1}^N a_i \phi_i(\uu_i(t), \theta_i) }^2.
% \end{align}
% Define $\hY(t) = \diag{ (y(t))} = \diag{([\yy_1(t); \dots; \yy_N(t)])} \in \R^{\dy \times \dy}$. Because $\sum_{i=1}^N ( \sum_{j \in \sy_i} [\hY(t)]_{:,j}) = y(t)$, we can re-write \eqref{eq:first_reformulation} as
% \begin{align}\label{eq:second_reformulation}
%    \rr(\theta) = \frac{1}{2}\sum_{t = 1}^T \norm{ \sum_{i=1}^N \big( \sum_{j \in \sy_i} [\hY(t)]_{:,j} -  a_i \phi_i(\uu_i(t), \theta_i) \big) }^2.
% \end{align}
% With the above set of reformulations, the system model identification problem is
% \begin{align}\label{eq:intermediate_opt_problem}
%    \hspace{-0.05in} \argmin_{\theta} \frac{1}{2}\sum_{t = 1}^T \norm{ \sum_{i=1}^N \big( a_i \phi_i(\uu_i(t), \theta_i) - \sum_{j \in \sy_i} [\hY(t)]_{:,j} \big) }^2.
% \end{align}
Optimization problem~\eqref{eq:intermediate_opt_problem} couples the parameters and data for all the agents. We next consider a reformulation described in~\cite{huang2022scalable} to set up a formulation for a distributed algorithm, allowing the sub-systems to do local computations and communicate with the neighboring sub-systems in the network $\G(\V,\E)$ to determine a solution for problem \eqref{eq:intermediate_opt_problem}.  Assuming the network $\G(\V,\E)$ is connected, let $\p \in \R^{N \times N}$ be a finite weight matrix associated with the graph satisfying the following assumption:
% \begin{align} 
%     \p_{ij} &= \begin{cases}\label{eq:pmat}
%     p_{ij} \in , & \text{if} \ i = j \ \text{or} \ (i,j) \in \E, \\
%     0, &  \text{otherwise}. 
% \end{cases}
% \end{align}
% We make the following assumption on matrix $\p$: 
\begin{assmp}\label{assmp:null_space}
  $\nul\{ \p \} = \sn \{\ones_N \}$.   
\end{assmp}
\noindent A few examples of matrices $\p$ that satisfy Assumption~\ref{assmp:null_space} are:
\begin{itemize}
       \item[(i)] Laplacian matrix: The Laplacian matrix of the graph \cite{merris1994laplacian} is defined as:
           \begin{align*} 
               \p_{ij} &= \begin{cases}
               -1, & \text{if} \ (i,j) \in \E, \\
               |\N_i|, & \text{if} \ i = j, \\
               0, &  \text{otherwise}. 
               \end{cases}
           \end{align*}
           % \item[(ii)]  Symmetric doubly stochastic matrix \cite{tsitsiklis1984problems} with the properties: $\p_{ij} \in [0,1], \p = \p^\top, \p \ones_N = \ones_N$. 
           \item [(ii)] A matrix created via a column stochastic matrix: $\p = \I_N - \tilde{\p}^\top$, where, $\tilde{\p}$ is a column stochastic matrix such that $[\tilde{\p}] \in [0,1], \ones_N^\top \tilde{\p} = \ones_N$.
   \end{itemize}

\noindent Let $\Phi(\uu(t), \theta) := [\phi_1(\uu_1(t), \theta_1); \phi_2(\uu_2(t), \theta_2); \dots; $ $ \phi_N(\uu_N(t), \theta_N)] \in \R^{N\dy}, \widehat{\yb}(t) := [ \widehat{y}(t); \widehat{y}(t); \dots; \widehat{y}(t)] \in \R^{N\dy}, \widehat{\p}:= \p \otimes \I_{\dy}, \x = [\theta; w] \in \mathbb{R}^{\mathbf{d}_\theta + N\dy}$.
Consider the problem
\begin{align}\label{eq:modified_prob}
      \hspace{-0.06in} \argmin_{\x = [\theta; w] } \sum_{t = 1}^T \left[ f_t(\x) := \textstyle \frac{1}{2N^2} \norm{ \Phi(\uu(t), \theta) -  \widehat{\yb}(t) - \widehat{\p} w }^2 \right ].
\end{align}
Let $F(\x) := \sum_{t = 1}^T f_t(\x)$ for brevity of notation. Problem~\eqref{eq:modified_prob} allows for the objective function $f_t$ to be distributed across different sub-systems and allows for the synthesis of a distributed algorithm. In the next result, we establish that solving problem~\eqref{eq:modified_prob} indeed aids towards our objective of solving~\eqref{eq:original_problem}. Specifically, Lemma~\ref{lem:optimality_equi} shows that the  solution to~\eqref{eq:modified_prob} is also a solution for~\eqref{eq:intermediate_opt_problem} and thus provides a solution for~\eqref{eq:original_problem}. We make the following assumption,
\begin{assmp}\label{assmp:solution_set}
    The set of minimizing solution to problem~\eqref{eq:modified_prob}, $\argmin_{\x} F(\x)$, is non-empty and bounded. 
\end{assmp}

\begin{lemma}\longthmtitle{Optimal solutions of~\eqref{eq:intermediate_opt_problem} and~\eqref{eq:modified_prob}}
\label{lem:optimality_equi}
Let the matrix $\widehat{\p}$ in~\eqref{eq:modified_prob} be such that $\p$ satisfy Assumption~\ref{assmp:null_space}. If $\x^\star = [\theta^\star ; w^\star]$ is a solution to problem~\eqref{eq:modified_prob}, then $\theta^\star$ is also a solution to problem~\eqref{eq:intermediate_opt_problem}.
\end{lemma}
\begin{proof}
Using the first order optimality conditions for~\eqref{eq:modified_prob} with convex function $F$ (sum of composition of convex and increasing functions), we have $\nabla_\theta F(\x^\star) = 0,\; 
    \nabla_w F(\x^\star) = 0$. Namely, for all $t \in \{ 1,2,\dots,T\}$,
\begin{align} 
    \nabla_\theta \Phi(\uu(t), \theta^\star)^\top  (\Phi(\uu(t), \theta^\star) -  \widehat{\yb}(t) - \widehat{\p} w^\star) & = 0, \label{eq:opt_mod_theta}\\
    \widehat{\p}^\top (\Phi(\uu(t), \theta^\star) -  \widehat{\yb}(t) - \widehat{\p} w^\star) & = 0. \label{eq:opt_mod_w}
\end{align}
Since, $\nul\{ \p \} = \sn \{\ones_N \}$, it follows from \eqref{eq:opt_mod_w} that there exists $z^\star$ such that 
\begin{align}\label{eq:dummy1}
    \ones_N \otimes z^\star :=  (\Phi(\uu(t), \theta^\star) -  \widehat{\yb}(t) - \widehat{\p} w^\star).
\end{align}
Multiplying both by $\ones_N^\top \otimes \I_\dy$ we get,
\begin{align*}
    N z^\star & = (\ones_N^\top \otimes \I_\dy) \ones_N \otimes z^\star \\
    & =  (\ones_N^\top \otimes \I_\dy) (\Phi(\uu(t), \theta^\star) -  \widehat{\yb}(t) - \widehat{\p} w^\star) \\ 
    & \textstyle \hspace{-0.3in} = \sum_{i=1}^N \big( \phi_i(\uu_i(t), \theta_i^\star)  - \widehat{y}(t) \big)  - (\ones_N^\top \otimes \I_\dy) (\p \otimes \I_{\dy}) w^\star \\
    &  \textstyle \hspace{-0.3in} = \sum_{i=1}^N \big( \phi_i(\uu_i(t), \theta_i^\star)  - \widehat{y}(t) \big) - (\ones_N^\top \otimes \p) (\I_\dy \otimes \I_{\dy}) w^\star \\
    &  \textstyle \hspace{-0.3in} = \sum_{i=1}^N \big( \phi_i(\uu_i(t), \theta_i^\star)  - \widehat{y}(t) \big),
\end{align*}
where we used Assumption~\ref{assmp:null_space} in the last step. Thus, 
\begin{align}\label{eq:dummy2}
    \textstyle z^\star = \frac{1}{N} \sum_{i=1}^N \big( \phi_i(\uu_i(t), \theta_i^\star)  - \widehat{y}(t) \big).
\end{align}
Substituting~\eqref{eq:dummy1} and \eqref{eq:dummy2} in \eqref{eq:opt_mod_theta}, gives: for all $t \in \{1,\dots,T\}$
\begin{align}\label{eq:dummy3}
    & \hspace{-0.08in} \nabla_{\theta_j} \textstyle \Phi(\uu_j(t), \theta^\star_j)^\top \Big[\sum_{i=1}^N \big( \phi_i(\uu_i(t), \theta_i^\star) \nonumber \\ 
    & \hspace{1in} - \widehat{y}(t) \big) \Big] = 0, \;\forall j \in \{1,2,\dots,N\}.
\end{align}
As problem~\eqref{eq:intermediate_opt_problem} is convex, by the optimality conditions, any $\theta^{'}$ is a solution of \eqref{eq:intermediate_opt_problem} if and only if, for all $t \in \{1,\dots,T\}$,
% \label{eq:dummy4}
\begin{align*}
    & \hspace{-0.12in} \nabla_{\theta_j} \phi_j \textstyle (\uu_j(t),\theta^{'}_j)^\top \Big[\sum_{i=1}^N \big(\phi_i(\uu_i(t), \theta_i^{'}) \nonumber \\
    & \hspace{1in} \textstyle - \widehat{y}(t) \big) \Big] = 0, \ \forall j \in \{1,2, \dots, N\}.
\end{align*}
Thus, we conclude $\theta^\star$ is a solution to~\eqref{eq:intermediate_opt_problem}.  
\end{proof}
Using Lemma~\ref{lem:optimality_equi}, we can concentrate on solving~\eqref{eq:modified_prob}, which facilitates the development of distributed solutions, as will be demonstrated in the following section.

\section{Online Model Update Algorithm}\label{sec:sysid_alg}
 % \subsection{Model identification via online experiments} \label{sec:onlinesysid}

Recall the problem statement in Section~\ref{sec:dis_model_ID}. Assume a local controller is available at each sub-system $i$ that solves problem~\eqref{eq:control_avg_model}. Starting at any time $t_0$ the local controller has access to the model with parameter $\theta^0$ that is used to estimate the output $\widehat{y}_{\theta^0}(t_0)$ for $[t_0, t_0 + T-1]$ based on which it generates local control decisions $\uu_i$ that are implemented in the system by the local actuators at some time $t_1 \geq t_0 + \Delta t$, where $\Delta t$ is the amount of time it takes to solve problem~\eqref{eq:control_avg_model}. At time $t_1$ a measurement $\widehat{y}(t_1)$ of the observable is obtained. If the model predicted output, $\widehat{y}_{\theta^0}(t_1)$, does not match the measurement, $\widehat{y}(t_1)$, the model parameter needs to be updated. We meet this objective by developing an online algorithm to solve the system model update problem~\eqref{eq:modified_prob}. The function $f_t (\x)$ in~\eqref{eq:modified_prob} is used to update the parameters to a new value at time $t$ given the measurement $\widehat{y}(t)$ and the control decisions $\uu_i(t)$. We aim to minimize the ``regret'' of the online algorithm compared to a model devised using all the input-output pairs in hindsight. Let $\{\x(t)\}_{t \geq 1}$ denote the solution parameters generated by our algorithm, we formally define regret of our algorithm after any time $T$ as,
\begin{align}\label{eq:regret}
    \rc_T := \sum_{t=1}^T f_t(\x(t)) - \min_{\x} \sum_{t=1}^T f_t(\x).
\end{align}
Note that if $\rc_T$ is zero, then the solution sequence $\{\x(t)\}_{t \geq 1}$ is such that the total error incurred is equal to the error obtained by minimizing the error objective function in~\eqref{eq:modified_prob} created by using the control decisions and measurement data over the entire time horizon $[t_0,t_0+T-1]$. We propose Algorithm~\ref{alg:onlinegd} to update $\x(t) = [\theta(t); w(t)]$ in an online manner.
\begin{algorithm}[ht]
    \SetKwBlock{Initialize}{Initialize:}{}
    \SetKwBlock{Input}{Input:}{}
    \SetKwBlock{Repeat}{For $ t = 0,1,2, \dots$}{}
    % \Input{ $\{\eta\}_{t\geq 1}$
    % }
    % \Initialize{$\x(0) = [ \theta_0; w_0] \in \R^{\dtheta + N\dy }$}
    \Repeat {
    - Given $\x(t) = [ \theta (t); w (t)] \in \R^{\dtheta + N\dy }, \uu(t), \widehat{y}(t)$ \\
    - $ \x(t+1) = \x(t) - \eta_t \nabla f_t (\x(t)) $ 
    }
    \caption{Online Input-Output Map Update}
    \label{alg:onlinegd}
\end{algorithm}

\noindent Lemma~\ref{lem:gradbound} establishes the boundedness of the gradient steps involved in Algorithm~\ref{alg:onlinegd}.
\begin{lemma}\label{lem:gradbound}
    Let Assumptions~\ref{assmp:bibo}-\ref{assmp:solution_set} hold. There exists constants $\eta_t > 0$ and $\delta < \infty$ such that $\|\nabla f_t \| := \left \|\left [ \begin{array}{cc}
         \nabla_\theta f_t \\
         \nabla_w f_t 
    \end{array} \right ] \right \| \leq \delta$ for all $t \in \{1,\dots,T\}$ with $\x(t)$ updated by Algorithm~\ref{alg:onlinegd}
\end{lemma}
\begin{proof}
    We start by presenting three supporting claims that we later utilize to prove the desired result.\\[0.5ex]   

\hfill\begin{minipage}{0.98\linewidth}
    \textit{Claim 1: Any $\gamma$ sub-level set $C_\gamma := \{ \x \ | \ f_t(\x) \leq \gamma \}$ of $f_t$ is bounded.} 
    
    \textit{Proof.} Given $\beta > 0$, let $v^\star \in \argmin_{\x} f_t(\x)$, with  $\|v^\star\| < \infty$. Define, $\Gamma_\beta := \{\x \ | \ \|\x - v^\star\| = \beta\} $ and $v_\beta = \inf_{\x \in \Gamma_\beta} f_t (\x)$. Note that $\Gamma_\beta$ is non-empty and compact. Since, $f_t$ is continuous, from the Weierstrass's theorem $v_\beta$ is attained at some point of $\Gamma_\beta$, we have $v_\beta > f_t (v^\star)$. For any $\x$ such that $\|\x - v^\star\| > \beta$, let $\alpha = \frac{\beta}{\|\x - v^\star\|}, \tilde{\x} = (1 - \alpha) v^\star + \alpha \x$. By convexity of $f_t$, we have
    \begin{align*}
        (1 - \alpha) f_t(v^\star) + \alpha f_t(\x) \geq f_t (\tilde{\x}).
    \end{align*}
Since $\| \tilde{\x} - v^\star\| = \alpha \| \x - v^\star\| = \beta$, $\tilde{\x} \in \Gamma_\beta$ and 
\begin{align*}
    f_t(\tilde{\x}) \geq v_\beta = \inf_{\x \in \Gamma_\beta} f_t(\x). 
\end{align*}
Combining the above two relations, we get
\begin{align*}
     f_t(\x) & \textstyle \geq \frac{f_t(\tilde{\x}) - f_t (v^\star)}{\alpha} + f_t(v^\star)  \geq  f_t(v^\star) + \frac{v_\beta - f_t (v^\star)}{\alpha}\\
     & = \textstyle f_t(v^\star) + \frac{v_\beta - f_t (v^\star)}{\beta} \|\x - v^\star\|.
\end{align*}
Because $v_\beta > f_t(v^\star)$ and $f_t(\x) \leq \gamma$, we derive
\begin{align*}
\|\x - v^\star\| \leq \textstyle \frac{\beta(\gamma - f_t(v^\star))}{v_\beta - f_t (v^\star)}.
\end{align*}
Thus, $\|\x - v^\star\| \leq \max \left \{ \beta, \frac{\beta(\gamma - f_t(v^\star))}{v_\beta - f_t (v^\star)} \right \}. \hspace{0.35in} \hspace*{\fill} \qed$
\end{minipage}
\vspace{0.1in}

\hfill\begin{minipage}{0.98\linewidth}
\textit{Claim 2: There exists a sufficiently small $\eta_t > 0$ such that for all $t \in \{1,2,\dots,T\}$, $\x(t+1)$ updated by Algorithm~\ref{alg:onlinegd} lies in the sub-level set $C_{f_t (\x(t))} := \{ \x \ | \ f_t(\x) \leq f_t (\x(t)) \} $.}

\textit{Proof.} By Taylor series expansion and $f_t \geq 0$,
\begin{align*}
    f_t(\x(t+1)) &= f_t \big(\x(t) - \eta_t \nabla f_t (\x(t))\big)\\
    & \hspace{-0.7in} = f_t(\x(t)) - \eta_t \|\nabla f_t (\x(t))\|^2 + o(\eta_t \nabla f_t (\x(t)))\\
    & \hspace{-0.7in} = \textstyle f_t(\x(t)) \! - \! \eta_t \left( \! \|\nabla f_t (\x(t))\|^2 \! + \! \frac{o(\eta_t \nabla f_t (\x(t)))}{\eta_t} \! \right) \! \\
    & \hspace{-0.7in} \leq f_t (\x(t)),
\end{align*}
for sufficiently small $\eta_t > 0$ by the definition of $o(\eta_t)$, which completes the proof. $\hspace*{\fill} \qed$
\end{minipage}
\vspace{0.1in}

\hfill\begin{minipage}{0.98\linewidth}
\textit{Claim 3: Let $\zb(t):= \Phi(\uu(t), \theta(t)) -  \widehat{\yb}(t) - \widehat{\p} w(t) $. Then, $\exists  \ \bar{\zb} <\infty$ such that $\|\zb(t)\| \leq \overline{\zb}$ for all $t \in \{1,2,\dots,T\}$.}

\textit{Proof.} Under Assumptions~\ref{assmp:bibo} and~\ref{assmp:separable_model}, 
\begin{align*}
    & \|\zb(t) - \zb(1)\| = \| \Phi(\uu(t), \theta(t)) -  \widehat{\yb}(t) - \widehat{\p} w(t) \ - \\
    & \hspace{1.2in} \Phi(\uu(1), \theta(1)) +  \widehat{\yb}(1) + \widehat{\p} w(1)\|\\
    & \leq \| \big(\Phi(\uu(t), \theta(t)) - \widehat{\p} w(t) \big) \ -  \big(\Phi(\uu(1), \theta(1)) - \widehat{\p} w(1)\big) \| \\
    & \hspace{0.2in} + \| \widehat{\yb}(t) -  \widehat{\yb}(1)\| \\
    & =  L_{m} \left \| \left [ \!\!\begin{array}{cc}
    \uu(t)  \\
    \x(t)
    \end{array} \!\! \right ]- \left [ \!\! \begin{array}{cc}
    \uu(1)  \\
    \x(1)
    \end{array} \!\! \right ]  \right \| + \| \widehat{\yb}(t) -  \widehat{\yb}(1)\|. 
\end{align*}   
Therefore, $\|\zb(t) - \zb(1)\| \leq 2 \overline{\yy} + 2 L_m (\overline{\uu} + D_m)$,
%\begin{align*}
%& \|\zb(t) - \zb(1)\| \leq 2 \overline{\yy} + 2 \max \{ L_{m}, L_m \|\widehat{\p}\|\} (\overline{\uu} + D_m).
%\end{align*}
where the results of Claims~{1} and~{2} are applied with $D_m : = \max_t Diam_{\|.\|}(C_{f_t(\x(t))}) $ and $L_m := \max_{1 \leq i \leq N} L_i$. Therefore, there exists a $\bar{\zb} < \infty$ that bounds $\|\zb(t)\|$ for all $t \in \{1,2,\dots,T\}.  \hspace*{\fill} \qed$ 
\end{minipage}\\[1ex]

\noindent With all the claims, we circle back to the proof of Lemma~\ref{lem:gradbound}. At any time index $t$, 
\begin{align*}
    \nabla f_t & = \left [ \begin{array}{cc}
         \nabla_\theta f_t \\
         \nabla_w f_t 
    \end{array} \right ] 
    %& \hspace{-0.35in} =   
    %\left [ \begin{array}{cc}
    %     \nabla_\theta \Phi(\uu(t), \theta(t))^\top  \big(\Phi(\uu(t), \theta(t)) -  \widehat{\yb}(t) - \widehat{\p} w(t) \big) \\
    %    -\widehat{\p}^\top \big(\Phi(\uu(t), \theta(t)) -  \widehat{\yb}(t) - \widehat{\p} w(t) \big) 
    %\end{array} \right ] \\
     =\left [ \begin{array}{cc}
         \nabla_\theta \Phi(\uu(t), \theta(t))^\top  \zb(t) \\
        -\widehat{\p}^\top \zb(t) 
    \end{array} \right ]. 
\end{align*}
Thus, $\|\nabla f_t\|^2 \leq \|\nabla_\theta \Phi(\uu(t), \theta(t))^\top  \zb(t) \|^2 + \|\widehat{\p}^\top \zb(t) \|^2 + 2 \|\nabla_\theta \Phi(\uu(t), \theta(t))^\top  \zb(t) \|\|\widehat{\p}^\top \zb(t) \| \leq (N L_m + \|\widehat{\p}\|)^2 \overline{\zb}^2$, where we used Assumption~\ref{assmp:separable_model} and claim~{3}. Hence, $\|\nabla f_t\| \leq (N L_m + \|\widehat{\p}\|) \overline{\zb} := \delta$. This completes the proof.
\end{proof}

\begin{theorem}\longthmtitle{Regret of Algorithm~\ref{alg:onlinegd}} \label{thm:gdonline_convg}
Let Assumptions~\ref{assmp:bibo}-\ref{assmp:solution_set} hold. Let $\x^\star \in \argmin_{\x} \sum_{t=1}^T f_t(\x)$ and $\eta_t = \frac{c_1}{\sqrt{t}}$, $c_1 > 0$. Then, the regret of Algorithm~\ref{alg:onlinegd} after any time $T$ is bounded. In particular,
\begin{align*}
    \rc_T = \frac{\delta_1 \sqrt{T}}{2} - \frac{\delta_2}{2} = O(\sqrt{T}), 
\end{align*}
where, $\delta_1$ and $\delta_2$ are some positive finite constants. Therefore, $\limsup_{T \to \infty} \rc_T/T \rightarrow 0$.
\end{theorem}
\begin{proof}
Let $\x^\star \in \argmin_{\x} \sum_{t=1}^T f_t(\x)$. Consider the update in Algorithm~\ref{alg:onlinegd}, $\x(t+1) - \x^\star = \x(t) - \eta_t \nabla f_t (\x(t)) - \x^\star$, then
\begin{align}
    \|\x(t+1) - \x^\star\|^2 &\leq \|\x(t) -  \x^\star\|^2 + \eta_t^2 \|\nabla f_t(\x(t))\|^2 \nonumber \\
    & \hspace{0.4in} - 2\eta_t \nabla f_t (\x(t))^\top (\x(t) - \x^\star). \nonumber
\end{align}
From Lemma~\ref{lem:gradbound}, there exists $\delta =: \sup_{t} \|\nabla f_t (\x(t))\| < \infty$,
\begin{align}
    &\|\x(t+1) - \x^\star\|^2 \leq \|\x(t) -  \x^\star\|^2 + \eta_t^2 \delta^2 \nonumber \\
    & \hspace{0.10in} - 2\eta_t \nabla f_t (\x(t))^\top (\x(t) - \x^\star) \nonumber \\
   &\Longrightarrow \nabla f_t (\x(t))^\top (\x(t) - \x^\star) \label{eq:gdonline_1} \\ &\hspace{0.3in}\leq \frac{\|\x(t) -  \x^\star\|^2 -  \|\x(t+1) - \x^\star\|^2}{2 \eta_t}  + \frac{\eta_t\delta^2}{2} \nonumber
\end{align}
By convexity of $f_t$, 
\begin{align}
    f_t(\x(t)) - f_t(\x^\star) \leq \nabla f_t (\x(t))^\top (\x(t) - \x^\star). \label{eq:gdonline_2}
\end{align}
Combining~\eqref{eq:gdonline_1} and~\eqref{eq:gdonline_2} gives
\begin{align*}
    f_t(\x(t)) \! - \! f_t(\x^\star) & \leq \frac{\|\x(t) \!- \! \x^\star\|^2 -  \|\x(t \! + \! 1) \! - \! \x^\star\|^2}{2 \eta_t}  + \frac{\eta_t\delta^2}{2}.
\end{align*}
Summing over $t = 1$ to $T$, 
\begin{align*}
    \rc_T & = \sum_{t=1}^T f_t(\x(t)) - \sum_{t=1}^T f_t(\x^\star) \\
    & \hspace{-0.2in} \leq \sum_{t=1}^T \left(\frac{\|\x(t) -  \x^\star\|^2 -  \|\x(t+1) - \x^\star\|^2}{2 \eta_t} \right) + \sum_{t=1}^T \frac{\eta_t\delta^2}{2} \\
    & \hspace{-0.2in} = \frac{\|\x(1) -  \x^\star\|^2}{2\eta_1} - \frac{\|\x(T+1) -  \x^\star\|^2}{2\eta_T} + \frac{\delta^2}{2} \sum_{t=1}^T \eta_t \\
    & + \frac{1}{2} \sum_{t = 2}^T \|\x(t) -  \x^\star\|^2 \left (\frac{1}{\eta_t} - \frac{1}{\eta_{t-1}} \right).
\end{align*}
From claims~1 and~2 in Lemma~\ref{lem:gradbound}, there exists $\Xi < \infty$ such that $\sup_{t} \|\x(t) - \x^\star\| \leq \sup_{t}  \|\x(t)\| + \|\x^\star\| \leq \Xi$. Therefore, 
\begin{align*}
    \rc_T & \textstyle \leq \frac{\Xi^2}{2} \left ( \frac{1}{\eta_1} +  \sum_{t = 2}^T \left (\frac{1}{\eta_t} - \frac{1}{\eta_{t-1}} \right) \right) + \frac{\delta^2}{2} \sum_{t=1}^T \eta_t \\
    &  = \textstyle \frac{\Xi^2}{2 \eta_T} + \frac{\delta^2}{2} \sum_{t=1}^T \eta_t.
\end{align*}
For $\eta_t = \frac{c_1}{\sqrt{t}}, \sum_{t=1}^T \eta_t = \sum_{t=1}^T \frac{c_1}{\sqrt{t}} \leq 1 + \int_{t=1}^T \frac{c_1}{\sqrt{t}} dt \leq 1 + [2c_1\sqrt{t}]_{1}^T \leq 2c_1\sqrt{T} + 1 - 2c_1$. Thus,
\begin{align*}
    \rc_T \leq \frac{(\Xi^2/c_1 + 2\delta^2c_1)\sqrt{T}}{2} - \frac{(2c_1-1)\delta^2}{2}.
\end{align*}
Therefore, $\limsup_{T \to \infty} \rc_T/T \rightarrow 0$.
\end{proof}

The result of Theorem~\ref{thm:gdonline_convg} establishes that Algorithm~\ref{alg:onlinegd} provides an estimated output close to the estimated output derived via a best-fixed model in hindsight and thus solves the adaptive model update problem. Next, we elucidate a methodology for implementing the Algorithm~\ref{alg:onlinegd} within a distributed framework. 

Up to this point, we have presented Algorithm~\ref{alg:onlinegd} to solve the model update problem under online experimental scenarios. In the following, we present how Algorithm~\ref{alg:onlinegd} can be implemented distributively. Consider a communication network $\G^c(\V^c,\E^c)$ with $ \V^c = \V \cup \{0\}, |\V^c| = N+1$, $\E^c = \E \cup \{(0,1), (0,2), \dots, (0,N)\} \subseteq (N+1) \times (N+1)$. The node index $0$ is the fusion center to which all the sensor measurements are relayed. There are two kinds of communication links in the graph $\G^c$ (a) $(i,j) \in \E$ with $i,j \in \{1,2,\dots, N\}$ and (b) $(0,j) \in \E^c$ with $j \in \{1,2,\dots, N\}$. The sub-systems communicate with each other via the link of kind (a) and the fusion center communicates with all the sub-systems via the communication links of the form (b). The fusion center communicates the measurement $\widehat{y}(t) \in \mathbb{R}^\dy$ to all the sub-systems $i \in \{1,2,\dots,N\}$. The updates in Algorithm~\ref{alg:onlinegd} utilizes $\nabla f_t (\x(t))$. From~\eqref{eq:opt_mod_theta} and~\eqref{eq:opt_mod_w}, we have 
\begin{align}\label{eq:for_grad_dist_impl}
    \nabla f_t (\x(t)) =\left [ \begin{array}{cc}
         \nabla_\theta \Phi(\uu(t), \theta(t))^\top  \zb(t) \\
        -\widehat{\p}^\top \zb(t)
    \end{array} \right ]. 
\end{align}
Note that $\zb(t)$ can be decomposed as, $\zb = [\zb_1(t); \zb_2(t); \dots; $ $\zb_N(t)]$, where for $i,j \in \{1,2,\dots,N\}$,
\begin{align*}
\zb_i(t) = \phi_i (\uu_i(t), \theta_i(t)) - \widehat{y}(t) - ( \widehat{\p}_{ii} w_{i}(t) + \sum_{j \in \N_i} \widehat{\p}_{ij} w_{j}(t)).     
\end{align*}
% \begin{align*}
% \zb_i(t) = a_i \phi_i (\uu_i(t), \theta_i) - \sum_{j \in \sy_i} [\hY(t)]_{:,j} - ( \widehat{\p}_{ii} w_{i} + \sum_{j \in \N_i} \widehat{\p}_{ji} w_{j}).     
% \end{align*}
A closer examination of~\eqref{eq:for_grad_dist_impl} yields that $\nabla f_t$ can be further written as, $ \nabla f_t = [(\nabla_1 f_t)^\top; (\nabla_2 f_t)^\top; \dots; (\nabla_N f_t)^\top]$, where
\begin{align}\label{eq:grad_dist_impl}
    \nabla_i f_t = \left [ \begin{array}{cc}
         \nabla_{\theta_i} \phi(\uu_i(t), \theta_i(t))^\top \zb_i(t) \\
         -\widehat{\p}_{ii} \zb_{i}(t) - \sum_{j \in \N_i} \widehat{\p}_{ji} \zb_{j}(t)
    \end{array} \right ],
\end{align}
% \begin{align}\label{eq:grad_dist_impl}
%     \nabla_i f_t = \left [ \begin{array}{cc}
%          a_i\nabla_{\theta_i} \Phi(\uu_i(t), \theta_i)^\top \zb_i(t) \\
%          -\widehat{\p}_{ii} \zb_{i}(t) - \sum_{j \in \N_i} \widehat{\p}_{ji} \zb_{j}(t)
%     \end{array} \right ],
% \end{align}
for all $i \in \{1,2,\dots,N\}.$ Thus, using~\eqref{eq:grad_dist_impl} the updates in Algorithm~\ref{alg:onlinegd} can be implemented in a distributed manner at any sub-system $i$ while maintaining an auxiliary variable $\zb_i$ as shown in Algorithm~\ref{alg:onlinegd_dist}. 

\begin{algorithm}[h!]
    \SetKwBlock{Initialize}{Initialize:}{}
    \SetKwBlock{Input}{Input:}{}
    \SetKwBlock{Repeat}{For $ t = 0,1,2, \dots$}{}
    \SetKwBlock{Agent}{For all $i \in \{1,2,\dots,N\}$}{}
    % \Input{ $\{\eta\}_{t\geq 1}$
    % } 
    % \Initialize{$\theta_i(1) \in \R^{\theta_i}, w_i(1) \in \R^{\dy}$ for all $i \in \{1,2,\dots,N\}$}
    \Repeat {
        % \Agent{ 
                \hspace{-0.12in} - Receive $\widehat{y}(t)$ from the fusion center\\
                \hspace{-0.12in} - Given $\theta_i(t) \in \R^{\theta_i}, w_i(t), w_j(t) \in \R^{\dy}, j \in \N_i$\\
                \hspace{-0.12in} - $\zb_i (t) = \phi_i (\uu_i(t), \theta_i(t)) - \widehat{y}(t) - \displaystyle \sum_{j \in \N_i \cup \{i\}} \widehat{\p}_{ij} w_{j}(t)$ \\
                \hspace{-0.12in} - $\theta_i(t+1) = \theta_i(t) - \eta_t \nabla_{\theta_i} \phi_i(\uu_i(t), \theta_i(t))^\top \zb_i(t)$ \\
                \hspace{-0.12in} - $w_i(t+1) = w_i(t) + \eta_t \displaystyle \sum_{j \in \N_i \cup \{i\}} \widehat{\p}_{ji} \zb_{j}(t)$
            % }
    }
    \caption{Distributed Online Input-Output Map Update at Sub-system $i$}
    \label{alg:onlinegd_dist}
\end{algorithm}
In Algorithm~\ref{alg:onlinegd_dist}, each agent $i$ engages in two rounds of communication on auxiliary variables $\zb_i$ and $w_i$. Importantly, the exchange of $\zb_i$ and $w_i$ among agents does not allow for the reconstruction of the model parameters $\theta_i$ or the local input data. As a result, the information transmitted across the communication network does not divulge any direct details regarding the sub-system's parameters or local data, thereby bolstering the privacy and security of the individual sub-systems.

\section{Controller with Online Model Updates}\label{sec:online_controller}
In this section, we provide an extension of the control problem~\eqref{eq:control_avg_model} beyond the fixed input-output map for the entire time horizon $T$. We consider the input-output map in~\eqref{eq:control_avg_model} to be frequently updated based on the latest control decision and measurement. The following formulation captures the time-varying aspects
\begin{align}\label{eq:control_online_model}
\minimize_{\uu(1), \dots, \uu(T) \in \mathbb{R}^\du} & \sum_{k=1}^T \sum_{i=1}^N \ell_i(\widehat{y}_{\theta_i(t)}(k),k) + h_i(\uu_i(k), k) \\
\mbox{subject to} \ & \uu_i(k) \in \mathcal{U}_i, \mbox{ for all} \ i = 1,2,\dots, N, \forall k, \nonumber  
\end{align}
where, $t$ is the model update counter. Whenever, the input-output map is updated the counter $t$ is increased by $1$. Given $t$, let $\tilde{\ell}_i (\widehat{y}_{\theta_i},T) := \sum_{k=1}^T \ell_i(\widehat{y}_{\theta_i(t)}(k),k), \tilde{h}_i (\uu_i,T) := \sum_{k=1}^T  h_i(\uu_i(k),k)$. While the input-output map is fixed the control decisions in problem~\eqref{eq:control_online_model} can be solved for via a projected gradient iteration given by, for all $i \in \{1,2,\dots,N\}$,
\begin{align}\label{eq:proj_gradient}
    & \uu_i(\tau + 1) =  \textstyle \Proj_{\mathcal{U}_i} \big\{\uu_i(\tau) - \\    
    &  - \alpha \big[\nabla_{\widehat{y}_{\theta_i}} \tilde{\ell}_i(\widehat{y}_{\theta_i(t)}(\tau),T)^\top\nabla_{\uu_i} \widehat{y}_{\theta_i} (\uu_i(\tau));
    \nabla_{\uu_i} \tilde{h}_i(\uu_i(\tau),T)\big] \big \}, \nonumber 
\end{align}
where $\tau$ is the iteration counter of the projected gradient steps and $\Proj_{\mathcal{U}_i}\{\cdot \}$ denote the projection operator. The control decisions obtained via~\eqref{eq:proj_gradient} are implemented in the system. Further, after $T_{con}$ consecutive iterations of the application of the control decisions to the system, the model update counter $t$ is incremented by $1$ and the input-output map is updated via Algorithm~\ref{alg:onlinegd_dist} utilizing the current measurement. Subsequently, the projected gradient iterations are reinitialized and performed with the new input-output map corresponding to the parameters $\theta_i(t+1)$. We summarize this in Algorithm~\ref{alg:online_model_and_control}.
\begin{algorithm}[h!]
    \SetKwBlock{Initialize}{Initialize:}{}
    \SetKwBlock{Input}{Input:}{}
    \SetKwBlock{Repeat}{Repeat}{}
    \SetKwBlock{Agent}{For all $i \in \{1,2,\dots,N\}$}{}
    % \Input{ $\{\eta\}_{t\geq 1}$
    % } 
    \Initialize{$t = \tau = 1, \uu_i(0), \theta_i(0), w_i(0)$ for all $i$}
    \Repeat {
        \uIf{$\mod(\tau,T_{con}) = 0$}{

            - Update the input-output map via Algorithm~\ref{alg:onlinegd_dist} \\
            - $ t = t+1$ \\
            - $\uu(\tau) = \uu(0)$
        }
        \uElse{
            - Compute $\uu(\tau)$ using~\eqref{eq:proj_gradient}\\
            - Apply $\uu(\tau)$ to the system\\
            
        }
    - $\tau = \tau + 1$    
    }
    \caption{Control with Online Input-Output Map Update for Solving Problem~\eqref{eq:control_online_model}}
    \label{alg:online_model_and_control}
\end{algorithm}

\section{Application Example and Numerical Simulations }\label{sec:num_sim}
In this section, we provide an example of an application of the proposed framework and present numerical results of a numerical simulation of the application example. 

\subsection{Power System Application Example}
Consider the problem of controlling and optimally managing the operation of a distribution power grid with penetration of photovoltaic energy sources (PES). We formulate this as a voltage regulation problem that fits the formulation in problem~\eqref{eq:control_online_model}. We assume that there are $N$ number of PES bus and they can adjust their power injections for voltage regulation. The control decision at PES bus $i$ at any time $k$ is given by $\uu_i(k):= [P_i(k); Q_i(k)] \in \mathbb{R}^2$, where, $P_i(k)$ and $Q_i(k)$ are the net active and reactive power injections at the PES bus $i$ at time $t$, respectively. 
%Active and reactive power injections at the PES buses are controlled to regulate the voltage magnitudes in the IEEE-$37$ bus system. 
The sets $\mathcal{U}_i := \{ [P_i;Q_i] : P^2_i + Q^2_i \leq S^2_{i,\mbox{max}}, 0 \leq P_i \leq \overline{P}_i \} $, where $S^2_{i,\mbox{max}}$ is the rated apparent power for the of the PES $i$ and $\overline{P}_i$ is the
maximum real power available with PES $i$. The input-output $\widehat{y}_{\theta(t)}(k) = \frac{1}{N} \sum_{i=1}^N  \phi_i(\uu_i(k), \theta_i(k))$ gives the mapping from power injections $\uu_i$ to the magnitudes of the voltages in the entire distribution grid. The functions $\ell_i(.)$ are designed to capture the engineering constraint of keeping the true voltages within the interval $[0.95,1.05]$. Quadratic functions $\uu_i(k) \to h_i(\uu_i(k),k)$ that penalize active power
curtailment and reactive power injections at the PES buses at time $k$ are chosen.

\begin{figure}[t] 
	\centering	
	\includegraphics[scale=0.75,trim={3.1cm 1.75cm 2.2cm 1.8cm},clip] {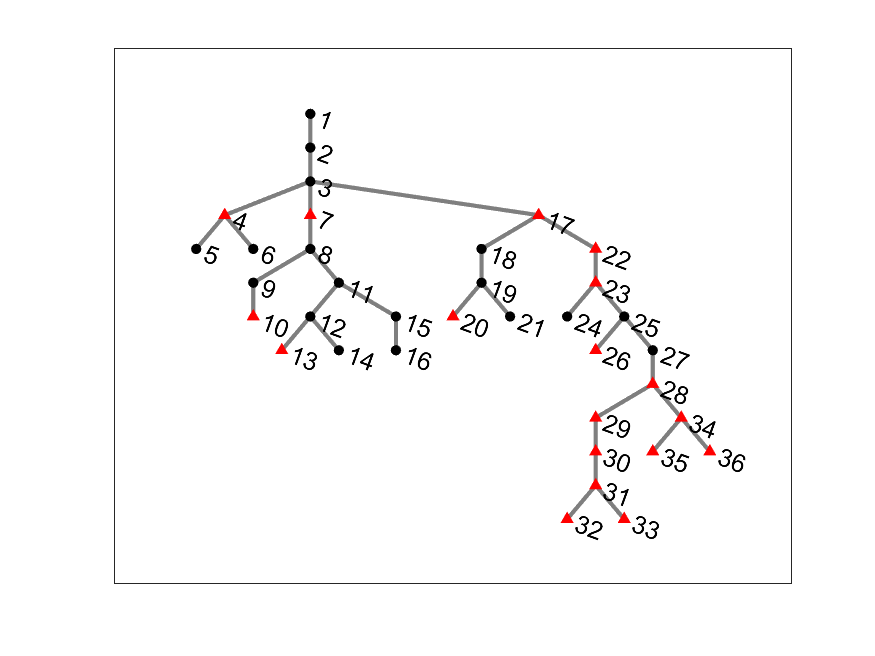}
	\caption{Schematic of the modified IEEE 37 bus system. The buses
		highlighted in red triangles are PES buses}
	\label{fig:modIEEE37} 
\end{figure}

% \begin{figure}[b] 
% 	\centering
% 	\includegraphics[scale=0.4,trim={0.2cm 0.9cm 0.2cm 0.1cm},clip] {images/online_without_control.eps}
% 	\caption{The voltage magnitudes ($p.u.$) over time in the modified IEEE $37$-bus system without any control.}	\label{fig:online_performance_without_control} 
% \end{figure}
\begin{figure}[b] 
	\centering
	\includegraphics[scale=0.4,trim={0.2cm 0.9cm 0.2cm -0.6cm},clip] {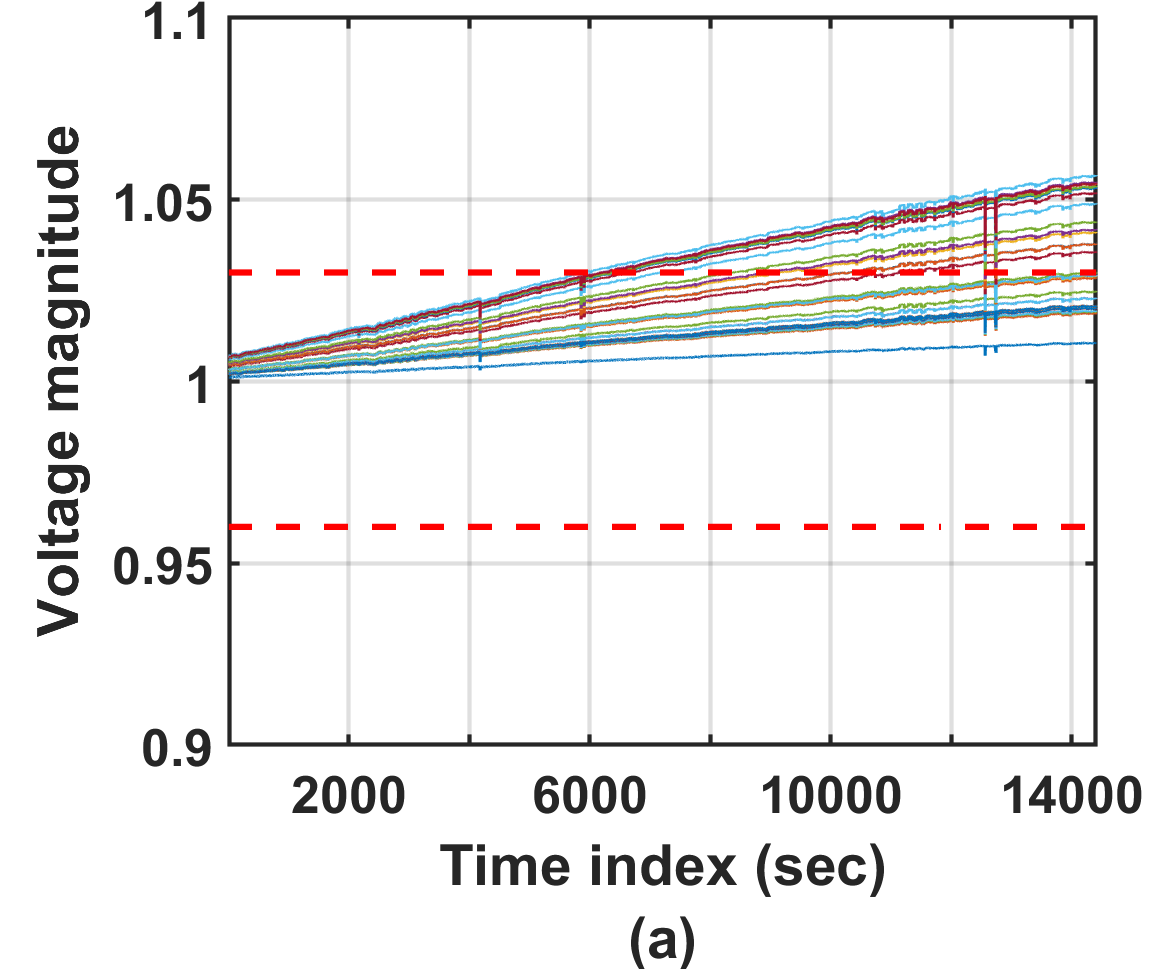}
	\caption{The voltage magnitudes ($p.u.$) over time in the modified IEEE $37$-bus system without any control.}	\label{fig:online_performance_without_control} 
\end{figure}

\subsection{Illustrative Numerical Simulations}
In this section, we instantiate the voltage regulation problem detailed in the previous section via a modified IEEE-$37$ bus system augmented with additional PES, using the solar irradiance data from Anatolia, CA, USA, and electric loads, having realistic load profiles for $4$ hours with a granularity of one second, introduced at different buses as illustrated in Fig.~\ref{fig:modIEEE37}. A total of $18$ PES buses are considered. We utilize Algorithm~\ref{alg:online_model_and_control} with $T_{con} = 1$ so that the input-output map between the power injections and the voltage magnitudes is updated every subsequent control decision. During the simulation study, PES exchange estimates through a connected communication network with $19$ nodes ($18$ PES and one fusion center). The graph Laplacian of the graph with $18$ nodes is used as the weight matrix $\p$. We compare two  parametric models for capturing the input-output map:
\begin{enumerate}
    \item A linear model where nodal active and reactive power injections serve as inputs and bus voltages as outputs, described by $\widehat{y}_{\theta_i(t)}(t) = A_i\uu_i(t)$ for all $t$. This model is known as the LinDistFlow model \cite{baran1989network}, and the goal is to identify matrix $A_i$.
    \item A non-linear model that posits a polynomial relationship between local power injections and bus voltage, described by $\widehat{y}_{\theta_i(t)}(t) = \frac{B_i - \sqrt{ B_i^2 - 4 (C_i - \overline{\uu}_i(t))}}{2}$ for all $t$, where $\overline{\uu}_i = \sqrt{P_i^2 + Q_i^2}$ is the apparent power magnitude at bus $i$. This reflects the quadratic correlation between power injection and voltage magnitude observed in power flow equations also known as the constant power load model. $B_i$ and $C_i$ are the parameters of the model to be determined.
\end{enumerate}

\noindent For Algorithm~\ref{alg:onlinegd_dist}, we adopt a step-size $\eta_t = 0.01/\sqrt{t}$. We plot the voltage levels in the modified IEEE $37$-bus system during the control process while updating the model with new input-output data becomes available using Algorithm~\ref{alg:online_model_and_control}. Fig.~\ref{fig:online_performance_without_control} illustrates potential violations of voltage regulation limits without control measures. The real-time estimated linear model provides a satisfactory control performance in Fig.~\ref{fig:online_performance_linear}, albeit with some fluctuations during periods of reduced PES generation. Notably, the control performance using the identified non-linear constant power load model, as shown in Fig.~\ref{fig:online_performance_nonlinear}, surpasses that of the linear model, which aligns with expectations given the non-linear model's closer representation of actual power flow dynamics. Overall, the models identified through our proposed algorithm demonstrate effective voltage regulation capabilities.

\begin{figure}[t] 
	\centering
	\includegraphics[scale=0.4,trim={0.2cm 1.5cm 0.2cm 0.1cm},clip] {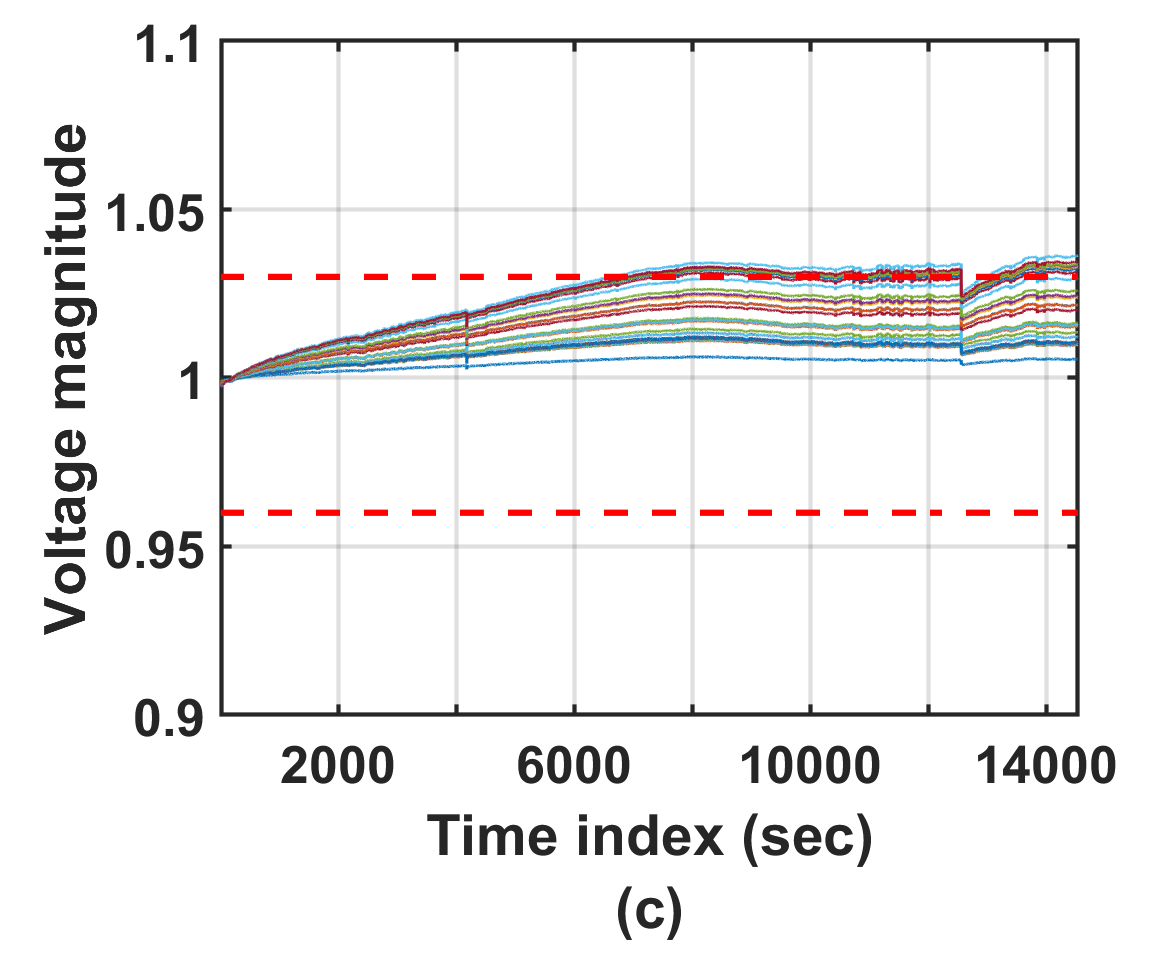} 
	\caption{The voltage magnitudes ($p.u.$) over time in the modified IEEE $37$-bus system with control decisions derived using the online distributed estimated linear map between power injections and voltage magnitudes. }	\label{fig:online_performance_linear} 
\end{figure}

% \begin{figure}[h] 
% 	\centering
% 	\includegraphics[scale=0.4,trim={0.2cm 1.4cm 0.2cm 0.1cm},clip] {images/online_gd_varying_etak.eps}
% 	\caption{The voltage magnitudes ($p.u.$) over time in the modified IEEE $37$-bus system with control decisions derived using a non-linear map between power injections and voltage magnitudes updated online. }	\label{fig:online_performance_nonlinear} 
% \end{figure}
\begin{figure}[h] 
	\centering
	\includegraphics[scale=0.4,trim={0.2cm 1.4cm 0.2cm -1cm},clip] {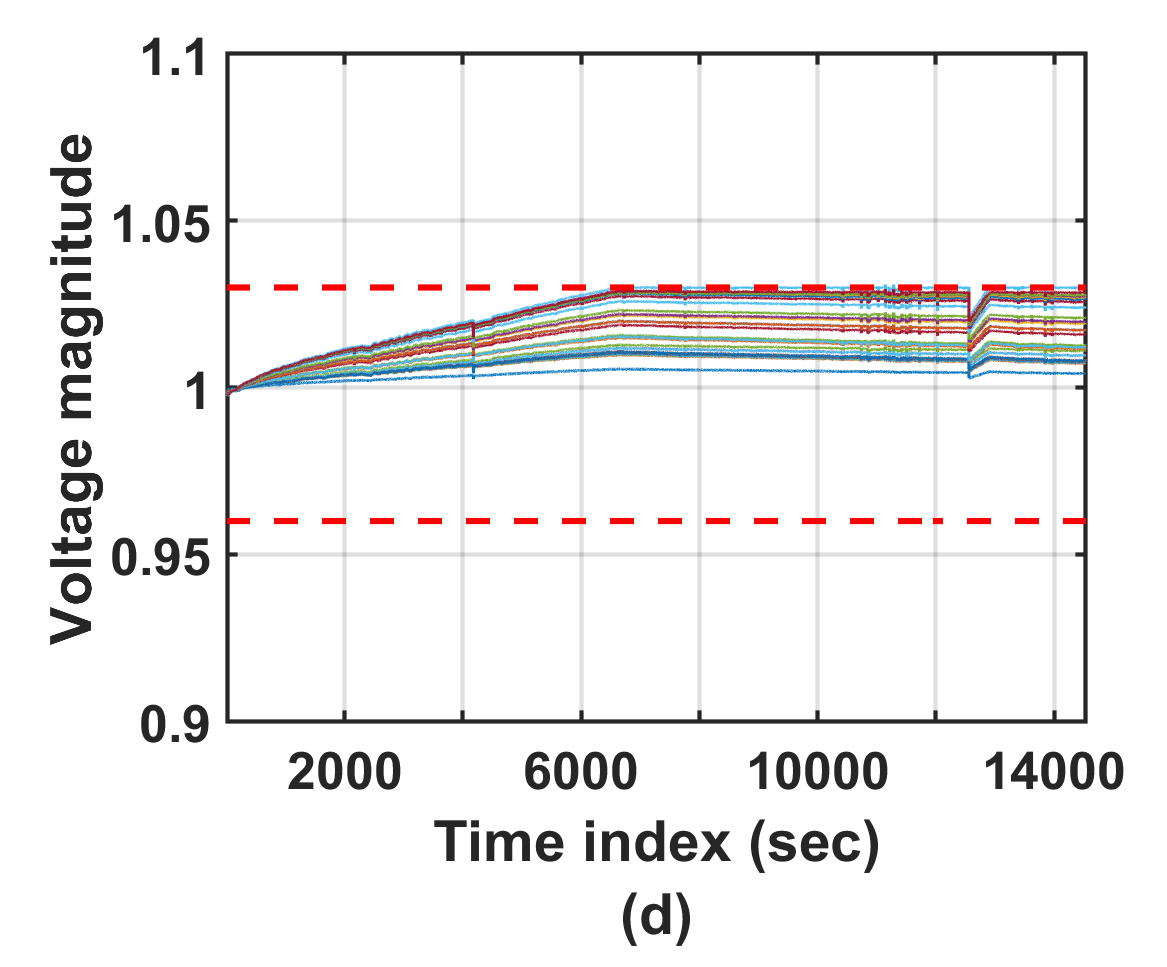}
	\caption{The voltage magnitudes ($p.u.$) over time in the modified IEEE $37$-bus system with control decisions derived using a non-linear map between power injections and voltage magnitudes updated online. }	\label{fig:online_performance_nonlinear} 
\end{figure}

\section{Conclusion}\label{sec:conclusion}
In this paper, we developed an online distributed algorithm where each agent updates its estimate of the model via an online gradient descent scheme utilizing the most recent input-output pair. We prove that the developed distributed algorithm has a sub-linear regret and determines the original system model. 
%The real-time updates of the agents, utilizing sequential data, significantly reduce the communication network’s bandwidth requirements. 
Further, agents only share non-linear estimates preserving their private information.  The numerical simulation study corroborates the efficacy of our developed algorithm with the identification of a more accurate quadratic power
flow model, which improves the voltage regulation performance of the control system. 
Looking ahead, our future endeavors will focus on extensive testing within real-world systems, aiming to further validate the performance and characterize the scalability of our method. 
with tens of thousands of nodes.

\bibliography{references}
% \begin{IEEEbiography}[
% {
% \includegraphics[width=1in,height=1.25in,clip,keepaspectratio]{images/Vivek_potrait.jpeg}
% }
% ]
% {Vivek Khatana} received the B.Tech degree in Electrical Engineering from the Indian Institute of Technology, Roorkee, in 2018. Currently, he is working towards a Ph.D. degree at the department of Electrical Engineering at University of Minnesota. His research interests include distributed optimization, consensus algorithms, distributed control and stochastic calculus, power system analysis and control.
% \end{IEEEbiography}

\end{document}